\documentclass[11pt]{article}
\usepackage[letterpaper, left=1in, right=1in, top=1in,bottom=1in]{geometry}
\PassOptionsToPackage{linesnumbered,noend}{algorithm2e} 
\PassOptionsToPackage{pagebackref=true}{hyperref}

\makeatletter
\renewcommand{\@biblabel}[1]{[#1]\hfill}
\makeatother

\usepackage[ruled,noend]{algorithm2e} 

\SetAlFnt{\small}
\SetAlCapFnt{\small}
\SetAlCapNameFnt{\small}
\SetAlCapHSkip{0pt}
\IncMargin{-\parindent}

\usepackage[utf8]{inputenc}

\usepackage{epsfig,amstext,xspace}
\usepackage[noend]{algpseudocode}

\usepackage{booktabs} 

\usepackage{geometry}
\usepackage[numbers]{natbib}

\usepackage{graphpap,amscd,mathrsfs,graphicx,lscape,dsfont,bm,url,subcaption}
\usepackage{epsfig,amstext,xspace}
\usepackage{mdframed}

\usepackage{comment}
\usepackage{amsthm}

\usepackage{auxiliary}

\usepackage[utf8]{inputenc} 
\usepackage[T1]{fontenc}    
\usepackage{url}            
\usepackage{amsfonts}       
\usepackage{nicefrac}       
\usepackage{microtype}      
\usepackage{bbm}
\usepackage{upgreek}

\numberwithin{equation}{section}

\newcommand{\remove}[1]{}

\usepackage{color}              
\usepackage[suppress]{color-edits}
\addauthor{tl}{cyan}
\addauthor{nd}{blue}
\addauthor{sc}{brown}

\usepackage{floatrow}
\newfloatcommand{capbtabbox}{table}[][\FBwidth]

\begin{document}

\title{Static pricing for multi-unit prophet inequalities}

 \author{
  Shuchi Chawla\thanks{The University of Texas at Austin, \texttt{shuchi@cs.utexas.edu}. Research supported by NSF grants CCF-2008006 and SHF-1704117.} \and
  Nikhil R. Devanur\thanks{Amazon.  \texttt{Iam@nikhildevanur.com}. Most of the work was conducted while the author was at Microsoft Research.}
  \and 
 Thodoris Lykouris\thanks{Massachusetts Institute of Technology,  \texttt{lykouris@mit.edu}.}  
 }
 
 \date{First version: July  2020\\Current version: June 2023%
\footnote{A $1$-page abstract of this work appeared in 17th Conference on Web and Internet Economics (WINE '21).}}

\maketitle

\begin{abstract}
We study a pricing problem where a seller has $k$ identical copies of a product, buyers arrive sequentially, and the seller prices the items aiming to maximize social welfare. When $k=1$, this is the so called \emph{prophet inequality} problem for which there is a simple pricing scheme achieving a competitive ratio of $1/2$. On the other end of the spectrum, as $k$ goes to infinity, the asymptotic performance of both static and adaptive pricing is well understood. 

We provide a static pricing scheme for the small-supply regime: where $k$ is small but larger than $1$. Prior to our work, the best competitive ratio known for this setting was the $1/2$ that follows from the single-unit prophet inequality. Our pricing scheme is easy to describe as well as practical -- it is anonymous, non-adaptive, and order-oblivious. We pick a single price that equalizes the expected fraction of items sold and the probability that the supply does not sell out before all customers are served; this price is then offered to each customer while supply lasts. This extends an approach introduced by Samuel-Cahn \cite{samuel1984comparison} for the case of $k=1$. This pricing scheme achieves a competitive ratio that increases gradually with the supply. Subsequent work \cite{JiangMaZhang_tightness} shows that our pricing scheme is the optimal static pricing for every value of $k$.
\end{abstract}
\addtocounter{page}{-1}
\thispagestyle{empty}

\newpage

\section{Introduction}
\label{sec:intro}
The \emph{prophet inequality} problem of Krengel and Sucheston
\cite{KrengelSuch1977} constitutes one of the cornerstones of online decision-making. A designer knows a set of $n$ distributions $\mathcal{F}_1,\ldots,\mathcal{F}_n$ from which random variables $X_t\sim \mathcal{F}_t$ are sequentially realized in an arbitrary order. Once a random variable is realized, the designer decides whether to accept it or not; \emph{at most one} realized random variable can be accepted. The objective is to maximize the value of the variable accepted, and the performance of the algorithm is evaluated against the \emph{ex-post} maximum realized. In a beautiful result, Samuel-Cahn~\cite{samuel1984comparison} showed that a simple static threshold policy achieves the optimal competitive ratio for this problem. Samuel-Cahn's algorithm determines a threshold $p$ such that the probability that there exists a realization exceeding the threshold is exactly $\nicefrac{1}{2}$, and then accepts the first random variable that exceeds the threshold. This algorithm achieves a competitive ratio of $\nicefrac{1}{2}$ against the ex-post optimum; no online algorithm, even one with adaptive thresholds, can obtain better performance.

Over the last few years, many extensions of the basic prophet inequality to more general feasibility constraints have been studied, and tight bounds on the competitive ratio have been established. However, one simple natural extension has largely been overlooked: where the designer is allowed to accept $k>1$ random variables for some small value of $k$. This is called the {\em multi-unit} prophet inequality. When $k$ is relatively large, then it is known that static threshold policies can achieve a competitive ratio of $1-O\Big(\sqrt{\frac{\log(k)}{k}}\Big)$ \cite{HadjiaghayiKleSan07} which goes to $1$ as $k\rightarrow \infty$. However, (for example,) for $k=2$ or $3$, prior to our work, the best known competitive ratio of static thresholds remained $\nicefrac{1}{2}$. Our work addresses this gap by posing and answering the following questions:
\vspace{0.1in}
\begin{center}
    \emph{Can a static threshold policy achieve a better competitive ratio than $\nicefrac{1}{2}$ for small $k=2,3,\ldots$? \\
    How should it be computed as a function of $k$? How does its performance scale with $k$?
    }
\end{center}
\vspace{0.1in}
\paragraph{The connection to mechanism design and static pricings.} A primary motivation for our work is its connection to welfare maximization in mechanism design. In this application, a seller has one or more units of an item to sell. The distributions correspond to known priors on the valuations of different customers (possibly heterogeneous), and the realizations correspond to the actual valuation of an incoming customer. The seller's goal is to maximize the social welfare, or the sum of values of the customers that obtain the item. Any online strategy for the prophet inequality problem corresponds to selecting prices for customers; customers buy if any units of the item are still available and their valuation is higher than the price. Static threshold policies correspond to static pricings, where the seller simply places a fixed price on the item and customers can purchase the item at that price while supply lasts. Static pricings have many nice properties that make them practical and suitable for real-world contexts. They are \emph{non-adaptive} (the price does not depend on which customers have already arrived) and \emph{order-oblivious} (the price does not depend on the order of customers). This makes their implementation simpler and removes the incentive on customers to strategize on the arrival order to obtain a better price, enhancing the customer experience.~\footnote{Although the order of customers does not affect the price assuming that the supply is not depleted, it \emph{does} affect the probability that the supply is depleted. This probability is $0$ for the first customer and increases as customers arrive.} Finally, static pricing is \emph{anonymous} (it does not discriminate based on which customer arrives), which is typically regarded as a more fair pricing scheme. We therefore focus on static pricings in this work.

\subsection{Our results} 
We answer the above questions by developing an algorithm for finding a static threshold policy for the multi-unit prophet inequality that is sensitive to the supply $k$. Our algorithm is very simple and practical. For any fixed price $p$, it estimates two statistics based on the given prior: (1) the fraction of items expected to be sold at that price, $\mu_k(p)$, illustrated in Fig.~\ref{fig:RUgraph1}, and (2) the probability that not all units will sell out before all the customers have been served, $\delta_k(p)$, illustrated in Fig.~\ref{fig:RUgraph2}. We then pick the static price $p^{\star}$
at which these two quantities are equal: $\mu_k(p^{\star})=\delta_k(p^{\star})$.

\begin{figure}[http]
\begin{center}

\begin{tikzpicture}
\def\X{6} \def\Y{4}
\tkzDefPoint(0,0){Origin}
\tkzDefPoint(0,\Y){Y}
\tkzDefMidPoint(Origin,Y) \tkzGetPoint{Y1}
\tkzDefPoint(\X,0){X}
\tkzDefMidPoint(Origin,X) \tkzGetPoint{X1}
\tkzDefPoint(0,0.9
*\Y){Zero_Price}
\tkzDefPoint(0.9*\X,0){Full_Price}
\tkzDefPoint(0,0.7*\Y){All}
\tkzDefPoint(0.5*\X,0){Price}
\tkzDefPoint(0,0.62*\Y){Y1}
\tkzDefPoint(0,0.15*\Y){Y2}
\tkzDefPoint(\X/3,0){X1}
\tkzDefPoint(2.2*\X/3,0){X2}
\tkzDefPoint(\X/3,0.62*\Y){P1}
\tkzDefPoint(0.5*\X,0.5*\Y){PricePoint}
\tkzDefPoint(0,0.5*\Y){muPoint}
\tkzDefPoint(0.5*\X,0.7*\Y){FullPoint}
\tkzDefPoint(2.2*\X/3,0.3*\Y){P2}
\tkzDrawSegments[arrows=->](Origin,X)
\tkzDrawSegments[arrows=->](Origin,Y)

\tkzLabelPoint[left](All){$k$}
\tkzLabelPoint[right](P2){$\textcolor{blue}{\Exp\Big[\min\big(\sum_{t=1}^n \textbf{1}\{X_t\geq p\},k\big)\Big]]}$}
\tkzLabelPoint[below](Price){$p^{\star}$}
\tkzLabelPoint[below](X){Price $p$}
\tkzLabelPoint[left](muPoint){$k\cdot\mu_k(p^{\star})$}

 \draw[line width=0.5mm, blue] (All) to[out=-10,in=160] (P1);
 \draw[line width=0.5mm, blue] (P1) to[out=-20,in=150] (PricePoint);
  \draw[line width=0.5mm, blue] (PricePoint) to[out=-30,in=120] (P2);
  \draw[line width=0.5mm, blue] (P2) to[out=-60,in=130] (Full_Price);
  
    \draw[red, line width=0.4mm] (FullPoint) to  (All);
     \draw[red, line width=0.4mm] (Price) to  (FullPoint);
     
    \draw[blue, dashed, line width=0.7mm] (muPoint) to  (PricePoint);
     \draw[blue,
     dashed, line width=0.7mm] (Price) to  (PricePoint);
     
    \draw[red,pattern=north east lines,pattern color=red] (0.5*\X,0) -- (0.5*\X,0.7*\Y)--(0,0.7*\Y);
    \draw[red,pattern=north east lines,pattern color=red] (0.5*\X,0) -- (0,0)--(0,0.7*\Y);
\end{tikzpicture}
\end{center}\caption{For a price $p^{\star}$, we cannot hope for more revenue than selling all units at $p^{\star}$, i.e., the red dashed area: $R_k(p^{\star})=kp^{\star}$. However, a price $p^{\star}$ thins the demand and results in an expected number of sales that is equal to $k\cdot \mu_k(p^{\star})
$ and thereby an expected revenue equal to the area below the blue dotted region: $\mu_k(p^{\star})\cdot R_k(p^{\star})$. }\label{fig:RUgraph1}
\end{figure}

\begin{figure}[http]
\begin{center}

\begin{tikzpicture}
\def\X{6} \def\Y{4}
\tkzDefPoint(0,0){Origin}
\tkzDefPoint(0,\Y){Y}
\tkzDefMidPoint(Origin,Y) \tkzGetPoint{Y1}
\tkzDefPoint(\X,0){X}
\tkzDefMidPoint(Origin,X)
\tkzDefPoint(0,0.7*\Y){All}

\tkzDefPoint(0.5*\X,0){Price}
\tkzDefPoint(0.5*\X,0.5*\Y){PricePoint}
\tkzDefPoint(0.8*\X,0.7*\Y){Full_Price_Point}

\tkzDrawSegments[arrows=->](Origin,X)
\tkzDrawSegments[arrows=->](Origin,Y)
\tkzLabelPoint[below](X){Price $p$}
\tkzLabelPoint[left](All){$1$}
\tkzLabelPoint[below](Price){$p^{\star}$}

 \draw[line width=0.5mm, olive] (Origin) to[out=30,in=240] (PricePoint);
 \draw[line width=0.5mm, olive] (PricePoint) to[out=60,in=180] (Full_Price_Point);
 \tkzLabelPoint[above](Full_Price_Point){$\textcolor{olive}{\Pr{\big[\sum_{i=1}^n \textbf{1}\{X_t\geq p\}}\leq k-1}\big]$}
 
\draw[olive, dashed, line width=0.7mm] (muPoint) to  (PricePoint);
\draw[olive,
     dashed, line width=0.7mm] (Price) to  (PricePoint);
\tkzLabelPoint[left](muPoint){$\delta_k(p^{\star})$}     

\end{tikzpicture}
\end{center}\caption{For a price $p^{\star}$, we cannot hope for more consumer surplus (utility) than the one of all customers with value higher than $p^{\star}$: $U(p^{\star}):=\sum_t \max(0,v_t-p^{\star})$. 
However, we do not receive utility from customer $t$ if there is no unit when she arrives; a lower bound on the probability of having a unit available is the probability $\delta_k(p^{\star})$ that not all units are depleted anyway. This lower bounds the welfare we collect from consumer surplus by $\delta_k(p^{\star})\cdot U(p^{\star})$.}\label{fig:RUgraph2}
\end{figure}

The competitive ratio of this static pricing increases gracefully as the supply increases and approaches $1$ at the rate of $1-O\Big(\sqrt{\frac{\log(k)}{k}}\Big)$ as $k\rightarrow \infty$, as we will explain next. The precise competitive ratio at any particular value of $k$ can be determined as the solution to a particular equation. Specifically, let $X$ be a Poisson random variable with a rate defined such that the following equation holds. 
\begin{equation}
    \frac{1}{k}\Exp[\min(X,k)]= \Pr[X\leq k-1].\label{eq:competitive_ratio_of_our_scheme}
\end{equation}
The worst-case competitive ratio of our algorithm is then given by the value of either side of the equation, say $\Pr[X\leq k-1]$. Note that this quantity is well defined because on the one hand, the truncated expectation $\frac{1}{k}\Exp[\min(X,k)]$ increases with the rate of the Poisson variable $X$ --- it is $0$ for rate equal to $0$ and $1$ for rate equal to infinity; on the other hand, the probability $\mathbb{P}[X\leq k-1]$ decreases with the rate --- it is $1$ when the rate is $0$ and $0$ when the rate is infinity. In effect, our analysis shows that the worst case for our static pricing occurs when the number of customers with value exceeding the price is given precisely by the Poisson variable $X$. As $k\rightarrow\infty$, the competitive ratio for this instance tends to $1-O\Big(\sqrt{\frac{\log(k)}{k}}\Big)$.

To obtain a better sense of the exact quantities the above equation leads to Figure~\ref{fig:compratio_supply_size} depicts the ratio as a function of $k$ and Table~\ref{tab:small_k} instantiates it for small values of $k$. 

\begin{figure}[htbp]
\begin{floatrow}
\ffigbox{%
\includegraphics[scale=0.3]{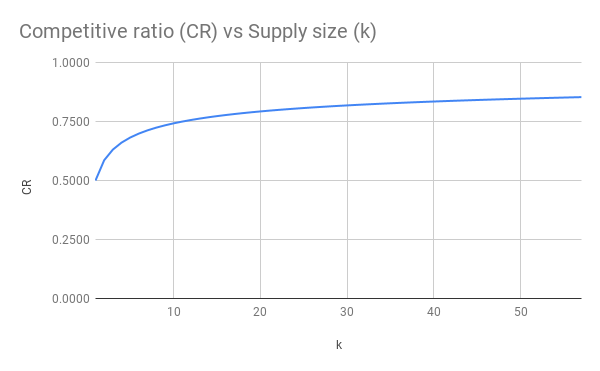}
}{%
  \caption{Competitive ratio of our static pricing as a function of the number of supply units $k$.}%
  \label{fig:compratio_supply_size}
}
\capbtabbox{%
  \begin{tabular}{|c|c|} \hline
  Number of units & Competitive ratio \\ \hline
  $k=1$ & $0.5$ \\\hline
  $k=2$ & $0.585$ \\\hline
  $k=3$ & $0.630$ \\\hline
  $k=4$ & $0.660$ \\  \hline  
  $k=5$ & $0.682$ \\\hline
  $k=6$ & $0.698$ \\\hline
  \end{tabular}
}{%
  \caption{Competitive ratio of our static pricing for small number of supply units~$k$}%
    \label{tab:small_k}
}
\end{floatrow}
\end{figure}

\subsection{Our techniques}

\paragraph{Samuel-Cahn's approach: Balancing revenue and utility contribution to welfare.}
Our methodology is inspired by the approach of Samuel-Cahn
\cite{samuel1984comparison} for the single-unit prophet inequality ($k=1$). First, the social welfare obtained by any static price $p$ can be expressed in two parts: (1) the expected revenue the seller obtains from selling any units of the item, and (2) the expected utility the buyers obtain from purchasing any units of the item. What is the most revenue and utility that we can expect at a particular price $p$? The most revenue the seller can obtain is simply the price $p$, $R(p):=p$. On the other hand, the most total utility the buyers can obtain at a price of $p$ is $U(p):=\sum_t \max(0,v_t-p)$, or the total excess value of the buyers above price $p$ assuming that {\em everyone who wants the item gets it}.~\footnote{In the pricing application, customers are assumed to have \emph{quasi-linear utilities}, i.e., they buy when their value is above the price and there is an available item. In this case, they obtain payoff equal to their value minus the price.}  It turns out that no matter what $p$ is, $R(p)+U(p)$ is an upper bound on the optimal-in-hindsight social welfare. Samuel-Cahn observed that when $k=1$, with the right choice of $p$ both the seller and the buyers can in expectation each obtain at least a half of these revenue and utility upper bounds respectively. One way to choose such a price $p$ is to ensure that the probability of selling the item is exactly $1/2$. At that price, on the one hand the buyer sells $1/2$ units in expectation, and on the other hand, each buyer has a probability at least $1/2$ of being offered the item and contributing to the total utility, resulting in the competitive ratio of $1/2$. Pictorially, the high-level idea behind Samuel-Cahn's approach is illustrated in Fig~\ref{fig:single_item}.

\begin{figure}[http]
\begin{center}

\begin{tikzpicture}
\def\X{6} \def\Y{4}
\tkzDefPoint(0,0){Origin}
\tkzDefPoint(0,\Y){Y}
\tkzDefMidPoint(Origin,Y) \tkzGetPoint{Y1}
\tkzDefPoint(\X,0){X}
\tkzDefMidPoint(Origin,X) \tkzGetPoint{X1}
\tkzDefPoint(0,0.9
*\Y){Zero_Price}
\tkzDefPoint(0.9*\X,0){Full_Price}
\tkzDefPoint(0,0.7*\Y){All}
\tkzDefPoint(0.5*\X,0){Price}
\tkzDefPoint(0,0.62*\Y){Y1}
\tkzDefPoint(0,0.15*\Y){Y2}
\tkzDefPoint(\X/3,0){X1}
\tkzDefPoint(2.2*\X/3,0){X2}
\tkzDefPoint(\X/3,0.62*\Y){P1}
\tkzDefPoint(0.5*\X,0.5*\Y){PricePoint}
\tkzDefPoint(0,0.5*\Y){muPoint}
\tkzDefPoint(0.5*\X,0.7*\Y){FullPoint}
\tkzDefPoint(2.2*\X/3,0.3*\Y){P2}
\tkzDefPoint(2*\X/3,0.1*\Y){Ustar}
\tkzDrawSegments[arrows=->](Origin,X)
\tkzDrawSegments[arrows=->](Origin,Y)

\tkzLabelPoint[left](All){$1$}
\tkzLabelPoint[right](P2){$\textcolor{blue}{\mu_1(p)=\Pr\big[\max{X_t}\geq p\}\big]}$}
\tkzLabelPoint[below](Price){$p^{\star}$}
\tkzLabelPoint[above](Ustar){$\textcolor{olive}{U(p^{\star})}$}
\tkzLabelPoint[below](X){Price $p$}
\tkzLabelPoint[left](muPoint){$\mu_1(p^{\star})=
\delta_1(p^{\star}) =\frac 12$}

 \draw[line width=0.5mm, blue] (All) to[out=-10,in=160] (P1);
 \draw[line width=0.5mm, blue] (P1) to[out=-20,in=150] (PricePoint);
  \draw[line width=0.5mm, blue] (PricePoint) to[out=-30,in=120] (P2);
  \draw[line width=0.5mm, blue] (P2) to[out=-60,in=130] (Full_Price);
  
    \draw[red, line width=0.4mm] (FullPoint) to  (All);
     \draw[red, line width=0.4mm] (Price) to  (FullPoint);
     
    \draw[blue, dashed, line width=0.7mm] (muPoint) to  (PricePoint);
     \draw[blue,
     dashed, line width=0.7mm] (Price) to  (PricePoint);

    \draw[red,pattern=north east lines,pattern color=red] (0.5*\X,0) -- (0.5*\X,0.7*\Y)--(0,0.7*\Y);
    \draw[red,pattern=north east lines,pattern color=red] (0.5*\X,0) -- (0,0)--(0,0.7*\Y);
    
\end{tikzpicture}
\end{center}\caption{The figure is similar to Fig.~\ref{fig:RUgraph1} as, for $k=1$, $\mu_1(p)=\Exp\Big[\min\big(\sum_{t=1}^n \textbf{1}\{X_t\geq p\},k\big)\Big]$ and the complement of Fig.~\ref{fig:RUgraph2} as $\delta_1(p)=\Pr{\big[\sum_{i=1}^n \textbf{1}\{X_t\geq p\}}\leq k-1\big]$. For $k=1$, $\mu_1(p)=1-\delta_1(p)$ for all prices and selecting $p^{\star}$ such that $\mu_1(p^{\star})=\delta_1(p^{\star})$ leads to competitive ratio of $1/2$. The root behind our algorithm is to extend this revenue-utility decomposition to $k>1$ where it no longer holds that $\mu_k(p)= 1- \delta_k(p)$ for all prices $p$.}\label{fig:single_item}
\end{figure}

\paragraph{Our extension to multiple units.} Extending this approach beyond a single unit, we similarly define $R_k(p):=pk$ to be the revenue obtained if all $k$ units of the item get sold at price $p$, and $U(p):=\sum_t \max(0,v_t-p)$ to be the total excess value of the buyers above price $p$ assuming that everyone who wants the item at price $p$ gets it. Then $R_k(p)+U(p)$ is an upper bound on the optimal in hindsight social welfare. Letting $\mu_k(p)$ denote the expected fraction of the supply sold at the price $p$, the seller's expected revenue is $\mu_k(p)R_k(p)$. On the other hand, the probability that a buyer is offered an unsold unit is at least as large as the probability that not all units are sold out at the end of the process; we call this probability $\delta_k(p)$. Then the total utility contributed by the buyers is at least $\delta_k(p)U(p)$. The static pricing $p$ therefore obtains at least a $\phi_k:=\min\{\mu_k(p),\delta_k(p)\}$ fraction of the upper bound $R_k(p)+U(p)$. Our pricing scheme
selects the price that maximizes this quantity $\phi_k$ . Since $\mu_k(p)$ is a decreasing function of $p$ and $\delta_k(p)$ is an increasing function, their minimum is maximized when the two are equal. Note that $\phi_k$ only depends on the buyers' value distributions and is independent of their order as both $\mu_k(p)$ and $\delta_k(p)$ are also order-oblivious quantities.

\paragraph{Crux of our analysis: Characterizing worst-case performance of our scheme.} 
The above description quantifies the competitive ratio of our scheme for any known distribution. To characterize its worst-case performance, we need to also identify worst-case distributions, i.e., those resulting to the lowest $\phi_k$. The crux of our analysis is a series of reductions eventually showing that Poisson distributions are these worst-case distributions. As a result, the competitive ratio of our scheme is $\phi_k$ for Poisson distributions and this leads to the competitive ratio we illustrated in \eqref{eq:competitive_ratio_of_our_scheme}.

\subsection{Related work}\label{ssec:related_work}
As already discussed, prophet inequalities were introduced by Krengel and Sucheston 
\cite{KrengelSuch1977}; Samuel-Cahn 
\cite{samuel1984comparison} provided a very clean analysis that our work builds upon. In the last decade, there has been a tremendous amount of work on extending prophet inequalities to different feasibility constraints over buyers (e.g., \cite{kleinberg2012matroid, duetting2015polymat, RubinsteinS17}), as well as to pricing with heterogeneous items where buyers have more complex valuations (e.g., \cite{feldman2014combinatorial, duetting2017, CDH+17, CMT19}). The reader is referred to \cite{Lucier_prophets} for a general survey.

\paragraph{Balanced prices.} The dominant approach for establishing prophet inequalities in combinatorial settings is by constructing so-called {\em balanced prices}, a technique introduced by Kleinberg and Weinberg
\cite{kleinberg2012matroid} and further developed in \cite{feldman2014combinatorial} and \cite{duetting2017}. This approach also has its roots in the work of Samuel-Cahn. Recall that the optimal social welfare is bounded by $R_k(p)+U(p)$ and that the pricing $p$ obtains a $\mu_k(p)$ fraction of the first term and a $\delta_k(p)$ fraction of the second. \cite{feldman2014combinatorial} choose a price $p_{1/2}$ that balances the revenue and utility upper bounds: $R_k(p_{1/2})=U(p_{1/2})$. They accordingly obtain a competitive ratio of $\frac{1}{2}(\mu_k(p_{1/2})+\delta_k(p_{1/2}))$. Noting that $\mu_k(p)+\delta_k(p)\ge 1$ at any price $p$, this competitive ratio is always at least $1/2$; but in general it is no better even when $k$ is large. In contrast, our approach picks a price where $\mu$ and $\delta$ are simultaneously larger than $1/2$. 

\begin{figure}[http]
\begin{center}

\begin{tikzpicture}
\def\X{6} \def\Y{4}
\tkzDefPoint(0,0){Origin}
\tkzDefPoint(0,\Y){Y}
\tkzDefMidPoint(Origin,Y) \tkzGetPoint{Y1}
\tkzDefPoint(\X,0){X}
\tkzDefMidPoint(Origin,X) \tkzGetPoint{X1}
\tkzDefPoint(0,0.9
*\Y){Zero_Price}
\tkzDefPoint(0.9*\X,0){Full_Price}
\tkzDefPoint(0,0.7*\Y){All}
\tkzDefPoint(0.5*\X,0){Price}
\tkzDefPoint(0.33*\X,0){NewPrice}
\tkzDefPoint(0,0.62*\Y){Y1}
\tkzDefPoint(0,0.15*\Y){Y2}
\tkzDefPoint(\X/3,0){X1}
\tkzDefPoint(2.2*\X/3,0){X2}
\tkzDefPoint(\X/3,0.62*\Y){P1}
\tkzDefPoint(0.5*\X,0.5*\Y){PricePoint}
\tkzDefPoint(0.33*\X,0.33*\Y){NewPricePoint}
\tkzDefPoint(0,0.5*\Y){muPoint}
\tkzDefPoint(0.5*\X,0.7*\Y){FullPoint}
\tkzDefPoint(0.33*\X,0.7*\Y){NewFullPoint}
\tkzDefPoint(2.2*\X/3,0.3*\Y){P2}
\tkzDrawSegments[arrows=->](Origin,X)
\tkzDrawSegments[arrows=->](Origin,Y)

\tkzLabelPoint[left](All){$1$}
\tkzLabelPoint[right](P2){$\textcolor{blue}{\mu_k(p)=\Exp\Big[\min\big(\sum_{t=1}^n \textbf{1}\{X_t\geq p\},k\big)\Big]]}$}
\tkzLabelPoint[below](Price){$p^{\star}$}
\tkzLabelPoint[below](NewPrice){$p_{1/2}$}
\tkzLabelPoint[below](X){Price $p$}
\tkzLabelPoint[left](muPoint){$\mu_k(p^{\star})$}

 \draw[line width=0.5mm, blue] (All) to[out=-10,in=160] (P1);
 \draw[line width=0.5mm, blue] (P1) to[out=-20,in=150] (PricePoint);
  \draw[line width=0.5mm, blue] (PricePoint) to[out=-30,in=120] (P2);
  \draw[line width=0.5mm, blue] (P2) to[out=-60,in=130] (Full_Price);
  
    \draw[red, line width=0.4mm] (NewFullPoint) to  (All);
     \draw[red, line width=0.4mm] (NewPrice) to  (NewFullPoint);
     
    \draw[blue, dashed, line width=0.7mm] (muPoint) to  (PricePoint);
     \draw[blue,
     dashed, line width=0.7mm] (Price) to  (PricePoint);
     
    \draw[red,pattern=north east lines,pattern color=red] (0.33*\X,0) -- (0.33*\X,0.7*\Y)--(0,0.7*\Y);
    \draw[red,pattern=north east lines,pattern color=red] (0.33*\X,0) -- (0,0)--(0,0.7*\Y);
    
     \draw[line width=0.5mm, olive] (Origin) to[out=30,in=240] (PricePoint);
 \draw[line width=0.5mm, olive] (PricePoint) to[out=60,in=180] (Full_Price_Point);
 \tkzLabelPoint[right](Full_Price_Point){$\textcolor{olive}{\delta_k(p)=\Pr{\big[\sum_{i=1}^n \textbf{1}\{X_t\geq p\}}\leq k-1}\big]$}
\end{tikzpicture}
\end{center}\caption{The price $p_{1/2}$ is chosen such that the shaded red area is half of the area under the solid blue curve. This price is in general smaller than the price $p^\star$ our approach chooses, and its worst case competitive ratio is $1/2$ for any $k$.}
\label{fig:RUgraph_balanced}
\end{figure}

\paragraph{Adaptive pricing and the magician's problem.} Multi-unit prophet inequalities were also previously studied by Alaei \cite{Alaei11} in the context of revenue optimal mechanism design. Alaei provided a competitive ratio of $\alpha_k:=1-\sqrt{\frac{1}{k+3}}$ for a more general problem called the magician's problem, which applies also to the multi-unit prophet inequality. 
Alaei's pricing scheme is not static in the sense of a price that you ``set it and forget it''. It uses a single price for all buyers, but probabilistically skips some buyers, in order to maintain a certain probability of not running out of items. This is equivalent to offering a price of $\infty$ to some users.
As such, it is more powerful than the class of static pricing policies that we consider,  and it does not satisfy two of the three properties mentioned above: non-adaptivity and order-obliviousness. The static pricing we develop provides a strictly better competitive ratio than that of Alaei for $k\in [2,20]$. Subsequent to our work, Jiang, Ma, and Zhang \cite{JiangMaZhang2021} improved upon Alaei's bound to obtain tight competitive ratios for the magician's problem at all values of $k$. This bound is again achieved by a dynamic pricing and therefore strictly exceeds the bound achieved by our static pricing.

\paragraph{Tightness of our pricing scheme.} Subsequently to our work, Jiang, Ma, and Zhang \cite{JiangMaZhang_tightness} established that our static pricing policy is worst-case optimal across all static threshold policies. In particular, \cite{JiangMaZhang_tightness} provide an instance where the gap between any static policy and the ex-post optimal is exactly $\phi_k$. We note that, for specific instances, one can obtain improved static prices. As a result, our policy is not  instance-optimal. For example, consider the policy that selects the price $p$ that directly maximizes $\mu_k(p)R_k(p)+\delta_k(p)U(p)$ --- such a policy would also enjoy the competitive ratio guarantees that we provide, while obtaining improved performance on some distributions. However, this policy requires the evaluation of $U(p)$ for all prices $p$ which makes it less simple (especially when one needs to estimate these quantities from data). In contrast, our policy only requires learning the price that equates $\mu_k(p)$ and $\delta_k(p)$, which can be done via binary search on the prices (as both can be evaluated with only samples from price $p$. Finally, although we present our results as comparing against the ex-post optimum, the same guarantees also hold against the corresponding ex-ante relaxation (see Remark~\ref{rem:exante}).

\paragraph{Other subsequent work.} In another subsequent work, Arnosti and Ma \cite{ArnostiMa2021} build on our techniques to study the performance of static threshold policies in the prophet-secretary setting with $k$ units, where customers arrive in uniformly random order (rather than in a worst-case order as what we consider). For this special case, Arnosti and Ma \cite{ArnostiMa2021} provide tight competitive ratio guarantees; in doing so, they heavily rely on the structure of our analysis and use the random order property to establish a better competitive ratio guarantee for the final step of our analysis. The performance of the worst-case distribution is the Poisson distribution, which is  similar to our result. 

\paragraph{Pricing with limited supply beyond prophet inequalities.} Our work lies in the general theme of providing supply-dependent guarantees  \emph{for pricing with known priors and limited supply}. Beyond prophet inequalities, such guarantees have also been provided in  ridesharing settings \citep{BanerjeeFreLyk17,BalseiroBrownChen19}. The latter works typically make a stronger assumption that the system is in steady state but has more complex state externalities: in multi-unit prophet inequalities, the supply just decreases when items are sold; in ridesharing it is reallocated across the network. To the best of our knowledge, these are the only two pricing settings where such supply-dependent guarantees with known priors and limited supply are provided; most prior work focuses on asymptotic optimality guarantees when the supply is large.

When the priors are not known in advance, a few other lines of work attempt to address these settings with the additional complication of learning information about the distributions. For example, dynamic pricing with limited supply has been studied in the context of prior-independent mechanisms, i.e., those that do not have distributional knowledge \cite{Babaioff12}; this work has been then extended in more general bandit settings under Knapsack constraints \citep{BadanidiyuruKleSli13,AgrawalDevanur14}. On the positive side, these approaches do not assume knowledge of the distributions; on the negative side, the guarantees they provide become meaningful only when the supply is large, e.g, $\sqrt{n}$ where $n$ is the number of buyers. 

\section{Model}
\label{sec:model}

\newcommand{\ft}{\mathcal{F}_t}
\newcommand{\f}{\mathcal{F}}
\newcommand{\welfare}{\textsc{Welfare}}
\newcommand{\vt}{v_t}
\newcommand{\compratio}{\textsc{CompRatio}} 
\newcommand{\worstcasecompratio}{\textsc{WorstCaseCompRatio}} 
\newcommand{\poffer}{{\bar{p}}}
\newcommand{\p}{\mathbf{p}}

An instance of the prophet inequality problem consists of
a set of $n$ distributions supported on non-negative real numbers
$\f=\{ \ft: t \in [n] \} $.\footnote{For ease of presentation, we denote by $\ft$ the $t$-th arriving distribution -- this order is not known to the seller.} In multi-unit prophet inequalities, there is also a supply $k$ that determines the number of units available for purchase at the beginning of time.

\paragraph{Static prices and pricing schemes.}
A static price is defined by a single number $p\in\R$.  The pricing works as follows.  Buyers arrive one by one and are offered a copy of the item at price $p$ as long as there is available supply.  Buyer $t$ has a value $\vt$ drawn independently from the distribution~$\ft$. The buyer purchases a unit of the item if and only if her value is above the price and there is an available item. In this case, the available supply decreases by $1$; otherwise the buyer leaves the system without an item and the available supply remains unaltered. 

A static pricing scheme $\boldsymbol{\pi}$ maps the supply $k$ and the distributions $\f$ to a static price $\pi(k,\f)\in\R$.

\paragraph{Performance metric.} The welfare of a static price $p\in \R$ 
on a particular realization of buyer values is the total value of the buyers who purchase a unit of the item. We denote its expected welfare by $\welfare(p,k,\f)$ where the expectation is over the randomness in buyer values drawn from $\f$. The benchmark we compare to is the expected optimal welfare in hindsight and is denoted by $\opt(k,\f)$, i.e., $\opt(k,\f)$ is the expected sum of the $k$ highest realized values drawn from the set of distributions $\f$. 
The \emph{competitive ratio} for a static pricing scheme $\boldsymbol{\pi}$ on supply $k$ is the worst-case welfare-to-optimum ratio across all the possible set of distributions $\f$, i.e.,
\begin{equation*}
    \compratio(\boldsymbol{\pi}
    ,k)=\inf_{\f}\frac{\welfare\big(\pi
    (k,\f),k,\f\big)}{\opt(k,\f)}
\end{equation*}
Our goal is to identify a static pricing scheme $\boldsymbol{\pi}$ that maximizes this worst-case competitive ratio. In the remainder of the presentation, we omit the arguments of $\pi(k,\f)$ when clear from the context. 

Without loss of generality, we assume that each distribution has $\Pr[v_t>0]>0$ (since otherwise we can ignore it) and further assume that $n>k$ (since otherwise setting a price of 0 is optimal).

\paragraph{Atomless assumption.} To ease the presentation of our scheme, we assume that the distributions are atomless.
Remark~\ref{rem:non_atomless} shows how our results extend to general distributions.  

\section{Our pricing scheme and its performance guarantee}
\paragraph{Decomposing to revenue and utility contributions.} For any fixed price $p\in\R$ and distributions $\f$, let $X_p$ denote the number of buyers who have value higher than the price. This is a random variable since the values of the buyers are drawn from the distributions $\f$; in particular it is equal to:
\begin{equation*}
    X_p:=|\{t:v_t\geq p\}|.
\end{equation*}
As in Samuel-Cahn's approach \cite{samuel1984comparison}, we decompose the welfare into two components: the total utility obtained by the buyers and the total revenue obtained by the seller. We now define some quantities of interest that determine these components. The first quantity is the probability that the seller runs out of units to sell, or in other words, that $X_p$ is at  least $k$. We use $\delta_k(X_p)$ to denote one minus this probability:
\begin{equation}\label{eq:delta_defn}
\delta_k(X):=\Pr\big[X\leq k-1\big].
\end{equation}
The second quantity is the expected fraction of units sold and is directly related to the revenue obtained by the seller. We use $\mu_k(X_p)$ to denote this truncated expectation:
\begin{equation}\label{eq:mu_defn}\mu_k(X):= \frac 1k \Exp\big[\min\{X,k\}\big].\end{equation}
The first important lemma that drives the design of our pricing scheme is that, for any distributions $\f$, the welfare-to-optimum ratio is at least the minimum of these two quantities.
\begin{lemma}\label{lem:any_price_minimum}
For any supply $k$, set of distributions $\f$, and any price $p\in\R$:
\begin{equation*}  \frac{\welfare\big(p,k,\f\big)}{\opt(k,\f)}\geq \min\big(\delta_k(X_p),\mu_k(X_p)\big).
\end{equation*}
\end{lemma}
\noindent This lemma is the main structural contribution of our work to the prophet inequality literature. Its proof is deferred to Section~\ref{sec:proof_lemma_any_price}.

\paragraph{Our pricing scheme.} For a given set of distributions $\f$ and supply $k$, our pricing scheme $\boldsymbol{\pi}$ outputs a static price $\pi(k,\f)=\pi$ that ensures that the two quantities in Lemma~\ref{lem:any_price_minimum} are equalized:
\begin{equation}\label{eq:delta_mu}
\delta_k\big(X_{\pi}\big)=\mu_k\big(X_{\pi}\big).
\end{equation}
The atomless assumption ensures that such a price always exists (see Remark~\ref{rem:non_atomless} on how the results extend beyond the atomless assumption). Observe that $\delta_k(X_p)$  is monotone non-decreasing in $p$ and $\mu(X_p)$ is monotone non-increasing in $p$. 
Moreover, $\mu_k(X_p)$ goes from $1$ to $0$ as the price goes from 0 to $\infty$ 
(since $n>k$). 
The atomless assumption ensures that both $\delta_k(X_p)$ and $\mu_k(X_p)$ are continuous. 
The intermediate value theorem then guarantees the existence of $\pi$.

We now define the competitive ratio of our pricing scheme for any distributions $\f$ and supply $k$:
\[\phi_k(\boldsymbol{\pi},\f):=\delta_k(X_{\pi(k,\f)})=\mu(X_{\pi(k,\f)}\big)\]
and the worst-case competitive ratio of $\p$ as
\begin{equation}\label{eq:def-phi-k}
\phi_k(\boldsymbol{\pi}):=\inf_{\f}\phi_k(\boldsymbol{\pi},\f)
\end{equation}
The second important lemma that enables our competitive ratio guarantee is that, for any fixed supply $k$, the minimum of $\delta_k(X_{\pi(k,\f)}\big)$ and $\mu_k(X_{\pi(k,\f)}\big)$ attains its lowest value when $\f$ consists of infinitely many Bernoulli random variables, all with equal bias; in this case, $X_p$ is a Poisson distribution. This is formalized in the following lemma.

\vspace{0.1in}
\begin{lemma}\label{lem:Poisson_worst_case}
For any supply $k$ and any set of distributions $\f$, $\phi_k(\boldsymbol{\pi},\f)$ attains its lowest value $\phi_k(\p)$ when  $\f$ is a collection of infinitely many Bernoulli distributions with equal bias, i.e., $X_{\pi(k,\f)}$ is a Poisson distribution.
\end{lemma}
The proof of this lemma stems from a series of reductions and is the main technical contribution of our analysis. Its proof is deferred to Section~\ref{sec:proof_Poisson_worst_case}.

\paragraph{Competitive ratio of our pricing scheme.}
The above two lemmas seamlessly establish the competitive ratio of our pricing scheme as demonstrated in the following theorem which is the main result of our work. The competitive ratio $\phi_k$ as a function of $k$ is illustrated in Figure~\ref{fig:compratio_supply_size}.
\vspace{0.1in}
\begin{theorem} Let $X^{\lambda}$ be a Poisson random variable with rate $\lambda$ and set $\lambda_k$ such that $\delta_k\big(X^{\lambda_k}\big)=\mu_k\big(X^{\lambda_k}\big)$. The competitive ratio of our pricing scheme $\boldsymbol{\pi}$ is at least $\phi_k:=\delta_k(X^{\lambda_k}) = \mu_k(X^{\lambda_k})$.
\end{theorem}
\begin{proof}
The proof of the theorem comes directly by combining Lemmas~\ref{lem:any_price_minimum} and \ref{lem:Poisson_worst_case}.
\end{proof}

\begin{remark}\label{rem:non_atomless}
If there are point masses in the distributions at price $\pi(k,\f)$,
we still obtain the same results provided we can break ties at random. 
A buyer with value $\pi(k,\f)$ is allocated the item with a probability such that  equality \eqref{eq:delta_mu} holds.
The definition of $X_{\pi(k,\f)}$ is adjusted accordingly: if $v_t = \pi(k,\f)$, then $t$ is counted only with some probability.  
The same effect can be achieved by randomly perturbing the price by an infinitessimal amount (although not static, this is still anonymous, non-adaptive and order-oblivious).
\end{remark}

\section{Welfare-to-optimum lower bound for any price (Lemma~\ref{lem:any_price_minimum})}\label{sec:proof_lemma_any_price}
This section proves Lemma~\ref{lem:any_price_minimum}: for any supply $k$, set of distributions $\f$, and any price $p\in\R$, we have
\begin{equation*}
    \frac{\welfare\big(p,k,\f\big)}{\opt(k,\f)}\geq \min\big(\delta_k(X_p),\mu_k(X_p)\big),
\end{equation*}
where
$\delta_k(X):=\Pr\big[X\leq k-1\big]$ and $\mu_k(X):= \Exp\big[\min\{X,k\}/k\big]$ as introduced in Eq.~\eqref{eq:delta_defn} and \eqref{eq:mu_defn}.

\begin{proof}[Proof of Lemma~\ref{lem:any_price_minimum}]
The proof follows the approach of Samuel-Cahn for the single-unit prophet inequality. We first bound the hindsight optimal welfare from above in terms of the price $p$ by bounding both the maximum possible revenue generated for the seller and the maximum possible utility generated for the buyers when posting price $p$. Let $Z_p$ denote the (random) set of buyers whose value exceeds the price $p$. Then we have:
\begin{align}
    \opt(k,\f) & = \Exp\left[\max_{S\subseteq [n]; |S|\le k} \sum_{t\in S} \vt\right]  \le \Exp\left[\max_{S\subseteq [n]; |S|\le k} \sum_{t\in S} (p+\max(0,\vt-p))\right]\nonumber \\
    & \le kp + \Exp\left[\sum_{t\in Z_p} (\vt-p)\right] \le kp + \sum_{t\in [n]} \Pr[t\in Z_p]\, \Exp\left[\vt-p | t\in Z_p\right]\label{eq:opt_upper_bound}.
\end{align}
We note that the first summand in the last term corresponds to what we referred in the introduction as $R_k(p)=kp$ while the second summand corresponds to the expected value of what we referred in the introduction as $U(p)=\sum_t \Exp[\max\big(0,v_t-p\big)]$.

We now decompose the expected welfare generated by price $p$ to a revenue and a utility component. The expected revenue of the seller upon setting price $p$ is:
\begin{equation} \textsc{Revenue}(p,k,\f)=\Exp\Big[\min\big(|Z_p|,k\big)\Big]\, p = \mu_k(X_p) kp. \label{eq:rev_bound}
\end{equation}
On the other hand, a buyer $t$ receives utility of $\vt-p$ if and only if: (1) $\vt$ is at least $p$, that is, $t\in Z_p$; and, (2) the item is still available when the buyer $t$ arrives. Regardless of the order in which buyers arrive, the latter event happens with probability at least as large as the probability that the item is not sold out at the end of the process. Recall that this latter probability is $\delta(X_p) = \Pr[X_p\le k-1]$. We therefore get the following lower bound on the utility generated by the pricing $p$:
\begin{equation}
\textsc{Utility}(p,k,\f)\geq \sum_{t\in [n]}
\delta_k(X_p)\cdot \Pr[t\in Z_p] \Exp\left[\vt-p | t\in Z_p\right]. \label{eq:util_bound}
\end{equation}
The proof of the lemma is completed by putting Eq. \eqref{eq:opt_upper_bound}, \eqref{eq:rev_bound}, and \eqref{eq:util_bound} together, which yields:
\begin{align*}
    \welfare(p,k,\f) &= \textsc{Revenue}(p,k,\f)+\textsc{Utility}(p,k,\f)\\ &\ge \mu_k(X_p) kp + \delta_k(X_p) \sum_{t\in [n]} \Pr[t\in Z_p] \Exp\left[\vt-p | t\in Z_p\right] \\
    & \ge \min\big(\mu_k(X_p), \delta_k(X_p)\big) \opt(k,\f).
\end{align*}
\end{proof}

\begin{remark}\label{rem:exante}
Although the result on Lemma~\ref{lem:any_price_minimum} bounds the expected welfare of any static price $p$ to the expected value of $\opt(k,\f)$, we note that its proof also implies that the same guarantee holds against the ex-ante relaxation.
\end{remark}

\section{Establishing Poisson as worst-case distribution (Lemma~\ref{lem:Poisson_worst_case})}\label{sec:proof_Poisson_worst_case}
This section proves Lemma~\ref{lem:Poisson_worst_case}. Recall that for any supply $k$ and any set of distributions $\f$, $\pi(k,\f)$ is  the price that  satisfies $\delta_k(X_{\pi(k,\f)})=\mu(X_{\pi(k,\f)}\big)$. We show that the corresponding competitive ratio $\phi_k(\boldsymbol{\pi})$ attains its lowest value when  $\f$ is a collection of infinitely many Bernoulli distributions with equal bias, i.e., $X_{\pi(k,\f)}$ is a Poisson distribution.

\begin{proof}[Proof sketch.]
To prove the above lemma, we progressively refine our understanding of the worst-case distributions, 
as outlined in the following 3 steps. 
\begin{enumerate}
    \item We reduce the problem of finding the worst distribution to a finite dimensional problem searching only over Bernoulli distributions (Section~\ref{ssec:reducing_to_Bernoulli}). Intuitively, our analysis is only affected by the probability that $v_t \geq \pi(k,\f)$ corresponding to the bias of a Bernoulli distribution. 
    \item  We show that all the Bernoulli biases are equal unless they are either 0 or 1 (Section~\ref{ssec:reducing_to_equal_bias}) . 
    \item We  show that the Bernoullis in fact must all have the same bias (Section~\ref{ssec:reducing_to_Poisson}).
\end{enumerate}
The lemma then follows by considering $n$ Bernoullis with the same bias and letting $n$ tend to infinity. The complete proof is provided at the end of the section.
\end{proof}

\subsection{Reducing to Bernoulli distributions}\label{ssec:reducing_to_Bernoulli}
\paragraph{Reducing worst-case distributions to Bernoulli distributions.}
A Bernoulli random variable with bias $b$ takes on the value $1$ with probability $b$ and $0$ otherwise.
We reduce the problem of finding the worst-case distribution to the following finite dimensional problem.
\begin{align}
    \phi^\star_k := \qquad & \min_{b_1, b_2, \ldots, b_n, \phi}  \phi \quad \text{ s.t. } \tag{$\textsc{GenBern}:\min \phi \textsc{ s.t. } {\delta=\mu=\phi}$}      \label{eq:bernoullioptimization_problem}\\ 
    & \quad X \text{ is the sum of } n \text{ Bernoullis with bias } b_1, b_2, \ldots, b_n \nonumber \\
    & \quad \delta_k(X) = \mu_k(X) = \phi.       \nonumber
\end{align}
    
\begin{lemma}
\label{lem:bernoullis}
For any supply $k>0$, number $n>k$ of customers, the worst-case competitive ratio~$\phi_k(\boldsymbol{\pi})$ of our pricing scheme is equal to the optimal value of~\eqref{eq:bernoullioptimization_problem}, $\phi^\star_k$.
\end{lemma}
\begin{proof}
$\phi_k(\boldsymbol{\pi})$ optimizes the objective of the problem \eqref{eq:bernoullioptimization_problem} across \emph{any} set of prior distributions while the optimization problem \eqref{eq:bernoullioptimization_problem} optimizes only over Bernoulli distributions. We  show that for every set of prior distributions, there exists a corresponding set of Bernoulli distributions that are feasible for \eqref{eq:bernoullioptimization_problem} and obtain the same objective function value.

The reduction is relatively simple. For any set of distributions $\f$ (not necessarily Bernoulli), we first compute the price $\pi(k,\f)$ of our pricing scheme, i.e., the one that makes $\delta_k\big(X_{\pi(k,\f)}\big)=\mu_k\big(X_{\pi(k,\f)}\big)$. Subsequently, for each distribution $\f_t$, we compute an equivalent Bernoulli bias $b_t=\Pr[v_t\geq \pi(k,\f)]$. The probability that any Bernoulli random variable is $1$ is therefore equal to the probability that its original counterpart is higher than the price $\pi(k,\f)$. As a result, both $\delta_k(\cdot)$ and $\mu_k(\cdot)$ are the same for the resulting sum as in the original problem which proves the lemma.
\end{proof}

\paragraph{A simpler equivalent way to express the resulting optimization problem.} We now define a slightly different form of the objective function which makes analyzing the optimal setting of the biases easier.  Here $\phi^\star_k$ is the optimal value of Problem \eqref{eq:bernoullioptimization_problem}:
\begin{align}
    & \min \mu_k(X) \text{ s.t. }       \tag{$\textsc{GenBern}:\min \mu \textsc{ s.t. } {\delta=\phi}$}      \label{eq:equivalentoptimization_problem}\\ 
    & \quad \delta_k(X) = \phi^\star_k       \nonumber\\
    & \quad X \text{ is the sum of Bernoulli r.v.s} \nonumber
\end{align}

\begin{lemma}\label{lem:equivalentOptProblem}
The optimal values of the optimization problems \eqref{eq:bernoullioptimization_problem} and \eqref{eq:equivalentoptimization_problem} are equal.
\end{lemma} 

\begin{proof}
The optimal solution $\phi_k^\star$ for the optimization problem \eqref{eq:bernoullioptimization_problem} is feasible for \eqref{eq:equivalentoptimization_problem} as the latter program needs to satisfy a subset of the former program's constraints; thus the optimum of \eqref{eq:equivalentoptimization_problem} is no larger than the one of \eqref{eq:bernoullioptimization_problem}.

For the opposite direction, assume that the optimum of \eqref{eq:equivalentoptimization_problem} is strictly smaller than the one of \eqref{eq:bernoullioptimization_problem}. Since $\delta_k(X)$ and $\mu_k(X)$ are both continuous, and monotone decreasing and increasing respectively, by increasing any of the biases, starting from the optimal solution of \eqref{eq:equivalentoptimization_problem},  we arrive to a new solution~$X'$, with $\delta_k(X')=\mu_k(X')<\delta_k(X)=\phi^{\star}_k$. This contradicts the fact that $\phi^{\star}_k$ was the optimum for \eqref{eq:bernoullioptimization_problem} and establishes that its optimal value is no larger than the one of \eqref{eq:equivalentoptimization_problem}.
\end{proof}

\subsection{Reducing to Bernoulli distributions with equal bias unless degenerate}\label{ssec:reducing_to_equal_bias}
We now show that the optimum of Problem \eqref{eq:equivalentoptimization_problem} is attained when all the Bernoulli distributions have either equal bias or are degenerate (with bias $0$ or $1$).

\vspace{0.1in}
\begin{lemma}\label{lem:equal_or_degenerate} 
The optimization problem \eqref{eq:equivalentoptimization_problem} is minimized when all non-degenerate Bernoulli distributions (that do not have bias $0$ or $1$) have equal bias.
\end{lemma}

We note that Hoeffding \cite{Hoeffding56} provides a similar result but without the constraint on $\delta_k(X)$. 
A generalization of Lemma~\ref{lem:equal_or_degenerate}  was obtained via a similar case analysis by subsequent work of Arnosti and Ma \cite{ArnostiMa2021} for the case where the objective and the constraint involve the expectation of an arbitrary non-negative integer valued function of $X$; in our case those functions are $\mu_k(X)$ and $\delta_k(X)$ respectively. Readers familiar with these results can safely skip this section and move to Section~\ref{ssec:reducing_to_Poisson}.

\paragraph{High-level structure of the reduction.} 
The key idea of the proof is to fix all but two of the biases and consider the problem of minimizing $\mu_k(X)$ subject to $\delta_k(X)$ being fixed as a function of these two biases. This is a problem in two dimensions, and we can characterize the optimal solutions to this problem. We then set aside these two distributions and assume by the principle of deferred decisions that they are instantiated in the end.
The eventual goal is to establish that $\mu_k(X)$ is minimized when these two biases are equal or degenerate (either $0$ or $1$). By working inductively on the number of biases that are not equal and are non-degenerate, we eventually establish that all biases should be equal or degenerate.

Formally, assume that we have $n$ Bernoulli distributions and fix all but two biases; let $b_1$ and $b_2$ be these two biases and refer to $r_1=1-b_1$ and $r_2=1-b_2$ as the rates of the respective random variables. Denote by $\bar{X}$ the sum of random variables drawn from the remaining $n-2$ distributions. Let $Q_1$ be the probability that $\bar{X}=k-1$ (equivalently, exactly one unit is left available for the last two distributions), $Q_2$ the probability that $\bar{X}=k-2$ (i.e., two units are left available), and $Q_{\ge 3}$ the probability that $\bar{X}\leq k-3$ (i.e. more than two units are left available). Finally, recall that $X$ is the sum of random variables drawn from all distributions. 

The following two claims enable the proof of Lemma~\ref{lem:equal_or_degenerate}.
\vspace{0.1in}
\begin{claim}
\label{claim:prog_1}
The problem of minimizing $\mu_k(X)$ as a function of $r_1$ and $r_2$ subject to $\delta_k(X)=\phi^{\star}_k$ for a constant value $\phi^{\star}_k$ is captured by the following program:
\begin{align}
     & \text{Maximize 
     } & \Big(Q_2+Q_{\ge 3}\Big)\cdot\Big(r_1+r_2\Big) + Q_1 r_1r_2
    \tag{Min-Revenue} \label{prog}\\
    & \text{subject to} & r_1, r_2 \in [0,1]^2 \quad \text{and} \quad r_1r_2\cdot \Big(Q_1-Q_2\Big)+ 
    \Big(r_1+r_2\Big)\cdot Q_2 & =\phi^\star_k-Q_{\ge 3}\label{eq:util_constraint_1}
\end{align} 
\end{claim}

\begin{claim}
\label{claim:prog-opt_2}
There always exists an optimal solution for the \eqref{prog} program that satisfies $r_1=r_2$, or $r_1\in\{0,1\}$, or $r_2\in\{0,1\}$.
Moreover, when $(Q_2)^2-Q_{\ge 3}(Q_1-Q_2)>0$, the unique optimal solution satisfies $r_1=r_2$.
\end{claim}

Using the above two claims, we can directly provide the proof of the lemma.

\begin{proof}[Proof of Lemma~\ref{lem:equal_or_degenerate}]
Among all the optimal solutions for Problem \eqref{eq:equivalentoptimization_problem}, consider the one that satisfies the following conditions. First, it has the fewest Bernoulli variables that are non-degenerate. Second, among those, it has the smallest difference between the largest non-degenerate bias and the smallest non-degenerate bias; call these biases $h$ and $s$ accordingly. Third, among those, it has the smallest number of variables with bias that is equal to either $h$ or $s$. We show, by contradiction, that $h=s$, establishing that all non-degenerate biases are equal. 

Among all the optimal solutions, we select the one satisfying the above criteria, and we select two distributions with bias $b_1=h$ and $b_2=s$ respectively. We apply Claim \ref{claim:prog_1} with these two and express \eqref{eq:equivalentoptimization_problem} as a function of $r_1=1-b_1$, $r_2=1-b_2$, and quantities $Q_1$, $Q_2$, $Q_{\geq 3}$ that are independent of $r_1$ and $r_2$. Claim~\ref{claim:prog-opt_2} establishes that there exists another optimal solution with all the other biases the same and the biases $b_1'$ and $b_2'$ corresponding to the two distributions either satisfy $b_1'=b_2'$, or have one them be degenerate. The latter is a contradiction as we assumed that the above solution has the smallest number of non-degenerate Bernoulli distributions. This means that $b_1'=b_2'$. However, in order for the new solution to induce the same $\delta_k(\cdot)$ and $\mu_k(\cdot)$, it means that $b_2\leq b_1'=b_2'\leq b_1$. Unless $b_1=b_2$, we have therefore identified a new optimal solution with fewer number of variables with bias equal to $h$ or $s$ which again would induce contradiction. As a result, $b_1=b_2$ and since $b_1=h$ and $b_2=s$, this means that $h=s$ and all non-degenerate biases are equal.\end{proof}

\paragraph{Expressing  \eqref{eq:equivalentoptimization_problem} as a function of two biases (Claim~\ref{claim:prog_1}).}
\begin{proof}[Proof of Claim~\ref{claim:prog_1}]
We first write the utility component $\delta_k(X)$, i.e., the probability that not all $k$ items are sold, in terms of the probabilities $Q_1, Q_2,$ and $Q_{\ge 3}$, and the biases $r_1$ and $r_2$:
\begin{align*}
    \delta_k(X) & = Q_{\ge 3}+Q_2(r_1+r_2-r_1r_2)+Q_1r_1r_2
\end{align*}
\noindent Applying this equation in the constraint of optimization probelm \eqref{eq:equivalentoptimization_problem} leads to the constraint \eqref{eq:util_constraint_1} of the \eqref{prog} optimization problem.

Our goal is now to minimize the objective of  \eqref{eq:equivalentoptimization_problem} which corresponds to the revenue component $\mu_k(X)$, i.e., the expected fraction of items sold. Since all but biases $b_1$ and $b_2$ are kept constant, this is equivalent to minimizing the number of items sold to the remaining two buyers. Our objective is therefore to minimize the following expression:
\[
\Big(r_1(1-r_2)+r_2(1-r_1)\Big)\cdot\Big(Q_1+Q_2+Q_{\ge 3}\Big)+\Big(1-r_1\Big)\cdot\Big(1-r_2\Big)\cdot \Big(Q_1+2Q_2+2Q_{\ge 3}\Big)
\]
The first term corresponds to the contribution of the two remaining buyers when only one of the two realized random variables is non-zero: then $\mu_k(X)$ is increased if there exists at least one item that is left available from the other buyers. The second term corresponds to the event that both buyers have non-zero realized random variables: in this case, $\mu_k(X)$ is increased by $1$ if exactly one available item is left from the other buyers and $2$ if at least two available items are left. Simplifying the expression leads to the following minimization objective:
\[
\Big(Q_1+2Q_2+2Q_{\ge 3}\Big) - \Big(Q_2+Q_{\ge 3}\Big)\cdot\Big(r_1+r_2\Big) - Q_1 r_1r_2
\]
Eliminating constant terms and negating the objective offers the maximization objective of the claim.\end{proof}

\paragraph{Auxiliary geometric interpretation of the constraint intersection with $[0,1]^2$.} We now provide a geometric fact about hyperbolas of a particular form which is useful in characterizing the optimal solutions of the \emph{Min-Revenue} program.
\vspace{0.1in}
\begin{claim} \label{claim:hyperbola_new}
Let $a,b\in\mathbb{R}$, and consider the hyperbola in the $x,y$ plane given by 
$xy + a(x+y) = b$. Expressing it as $y=g(x)$, and considering the segment(s) of $g$ that intersect the region $[0,1]^2$, the following holds:
\begin{itemize}
    \item when $a, b>0$, and $2a+1-b>0$, the segment of $g$ intersecting $[0,1]^2$  is convex and decreasing.
    \item when $a, b<0$, and $2a+1-b<0$, the segment of $g$ intersecting $[0,1]^2$  is concave and decreasing.
\end{itemize}
\end{claim} 
\begin{proof}
In both cases considered, $a^2+b>0$. In the first case, this is straightforward as $a,b>0$. In the second case, $2a+1<b$ and adding $a^2$, we obtain $a^2+2a+1<a^2+b$, implying $a^2+b>0$. If $x=-a$, then it means that $ax=b$ and hence $-a^2=b$; this contradicts the fact that $a^2+b>0$.

As a result, $x\neq -a$ and the hyperbola can be expressed as $y=\frac{b-ax}{x+a}$.
The first derivative of this function is 
$\frac{dy}{dx}=\frac{-a(x+a)-(b-ax)}{(x+a)^2}=\frac{(-b-a^2)}{(x+a)^2}$. This is negative as $a^2+b>0$ which establishes that $y$ is decreasing in $x$. The second derivative of this function is $\frac{d^2y}{dx^2}=-2\frac{(-b-a^2)}{(x+a)^3}=2\frac{a^2+b}{(x+a)^3}$. The convexity or concavity is determined by the sign of this derivative which is determined by the sign of $x+a$. Hence the resulting hyperbola has two segments, one convex for $x>-a$ and one concave for $x<-a$.
\begin{itemize}
    \item In the first case, since both $x\geq0$ and $a>0$, this sign is positive and the segment is convex.
    \item In the second case, because of the symmetry, if the convex segment intersects $[0,1]^2$, it should also intersect it at $x=y$. Hence, the point of this intersection is given by the equation $x^2+ax=b$ whose roots are $x=-a\pm \sqrt{a^2+b}$. The negative root $-a-\sqrt{a^2-b}<0$ and the positive root is $-a+\sqrt{a^2-b}>-a+\sqrt{a^2+2a+1}>1$ (since $2a+1<b$). As a result, the convex segment does not intersect with $[0,1]^2$ and, if the hyperbola does intersect, this happens with its concave segment.
\end{itemize}
\end{proof}

\paragraph{Characterizing the optimal solutions of the \emph{Min-Revenue} program (Claim~\ref{claim:prog-opt_2}).}
\begin{proof}[Proof of Claim~\ref{claim:prog-opt_2}]
We first start from two corner cases.
\begin{itemize} \item[{\em I(a):}] We start from the simplest case where $Q_2=0$. This transforms the \ref{prog} program to: 
\[\text{maximize} \quad Q_{\geq 3}(r_1+r_2) + Q_1 r_1r_2\quad \text{subject to} \quad Q_1r_1r_2=\phi^\star_k-Q_{\ge 3}.
\]
In this case, depending on the sign of $\phi^\star_k-Q_{\ge 3}$, the optimal solution either satisfies $r_1=r_2$ or $r_1\in\{0,1\}$ or $r_2\in\{0,1\}$ (if $Q_1\neq 0$) or is independent of $r_1,r_2$ (otherwise).
\item[\em{I(b):}] We then consider the case where $Q_1=Q_2>0$ (where it also holds that $(Q_2)^2-Q_{\ge 3}(Q_1-Q_2)> 0$); then the constraint \eqref{eq:util_constraint_1} is a linear constraint symmetric in $r_1$ and $r_2$ whereas the objective becomes equivalent to maximizing $r_1r_2$.  The optimum is therefore uniquely achieved at $r_1=r_2$.
\end{itemize}

\noindent For the remaining cases, we assume that $Q_2>0$ and $Q_1\neq Q_2$. We can therefore rewrite the objective in terms of $r_1+r_2$, obtaining the following equivalent formulation:
\begin{align*}
    & \text{Minimize } & (r_1+r_2)\cdot \frac{(Q_2)^2-Q_{\ge 3}(Q_1-Q_2)}{Q_1-Q_2}\\
    & \text{subject to} & r_1, r_2 \in [0,1]^2 \quad \text{and} \quad r_1r_2+
    (r_1+r_2) \cdot \frac{Q_2}{Q_1-Q_2} & =\frac{\phi^\star_k-Q_{\ge 3}}{Q_1-Q_2}
\end{align*}
The constraint can be rewritten as $r_1r_2+a(r_1+r_2)=b$ for $a=\frac{Q_2}{Q_1-Q_2}$ and $b=\frac{\phi^{\star}_k-Q_{\geq 3}}{Q_1-Q_2}$.

Observe that unless $r_1=r_2=1$ is the only feasible solution (in which case the claim is already proven), it holds that $Q_1-Q_2+2Q_2> r_1r_2(Q_1-Q_2) + (r_1+r_2)Q_2= \phi^{\star}_k-Q_{\geq 3}>0$.\footnote{Note that for any feasible solution we have $\phi^\star_k=\delta_k(X)\ge Q_2+Q_{\ge 3}>Q_{\geq 3}$. So, $\phi^\star_k-Q_{\geq 3}>0$.}
Hence, it holds that $a$, $b$, and $1+2a-b$ all have the same sign and this is the sign of $Q_1-Q_2$. By Claim~\ref{claim:hyperbola_new}, the hyperbola $r_2=g(r_1)$ corresponding to the constraint is always decreasing and is a) convex when $Q_1>Q_2$ ($a>0$) and b) concave when $Q_1<Q_2$ ($a<0$); it is also always symmetric in $r_1,r_2$.
\begin{itemize}
    \item[{\em II(a)}] If $Q_1<Q_2$, which implies $(Q_2)^2-Q_{\ge 3}(Q_1-Q_2)> 0$, the multiplier of the objective is negative in this case; therefore our goal is to maximize $r_1+r_2$ subject to a concave and decreasing constraint. As a result, the optimum is uniquely determined on the line $r_1=r_2$.
    \item[{\em II(b)}] If $Q_1> Q_2$ and $(Q_2)^2-Q_{\ge 3}(Q_1-Q_2)> 0$ the multiplier of $(r_1+r_2)$ in the objective is positive; therefore our goal is to minimize $(r_1+r_2)$ subject to a convex and decreasing constraint. As a result, the optimum is uniquely determined on the line $r_1=r_2$.
    \item[{\em II(c)}] If $Q_1> Q_2$ and $(Q_2)^2-Q_{\ge 3}(Q_1-Q_2)< 0$, the multiplier of $(r_1+r_2)$ in the objective is negative; therefore our goal is to maximize $(r_1+r_2)$ subject to a convex and decreasing constraint. As a result, the optimum lies on the boundary of $[0,1]^2$.
\end{itemize}
The last case is that $(Q_2)^2-Q_{\geq 3}(Q_1-Q_2)=0$ Then, if the program is feasible, there is a feasible solution with $r_1=r_2$ or with one of $r_1,r_2$ be in $\{0,1\}$ via a similar reasoning as above. 

Finally, $(Q_2)^2-Q_{\geq 3}(Q_1-Q_2)=0>0$ corresponds to the cases {\em I(b)}, {\em II(a)}, and {\em II(b)}. In all three cases, $r_1=r_2$ is the unique optimal solution. \end{proof}

\subsection{Reducing to the Poisson distribution}\label{ssec:reducing_to_Poisson}
Lemma~\ref{lem:equal_or_degenerate}  states that the optimal solution to \eqref{eq:equivalentoptimization_problem} is such that some of the biases are $0$ or $1$, and the remainder have equal bias $b_t=1-r$. We now show that in fact there cannot be any $0$ or $1$ biases. 
\vspace{0.1in}
\begin{lemma}\label{lem:equalbiases_new}
The optimization problem \eqref{eq:equivalentoptimization_problem} is minimized when all of the Bernoulli distributions have equal bias (and there are no degenerate distributions).
\end{lemma}
\begin{proof}
  By Lemma~\ref{lem:equal_or_degenerate}, we know that there exists an optimal solution that consists of Bernoulli distributions with bias $0$, $1$, or $1-r$ for a fixed $r$. Consider any such optimal solution. We first observe that this solution cannot have $k$ or more Bernoullis with bias $1$, otherwise $\delta_k(X)=0<\phi^\star_k$.~\footnote{This holds since $\phi_k^{\star}$ is at least $\nicefrac{1}{2}$ as it is the intended competitive ratio.}

Now, we pick two specific Bernoulli variables with unequal bias and reoptimize the objective over these, keeping the rest fixed. If there is a Bernoulli with bias $1$ and another with bias $1-r\ne 1$, we pick two such variables. Otherwise, we pick one variable with bias $1-r$ and another with bias $0$. We define $Q_1$, $Q_2$, and $Q_{\geq 3}$ as in Section~\ref{ssec:reducing_to_equal_bias}. In either case, we argue that $(Q_2)^2-Q_{\geq 3}(Q_1-Q_2)>0$. Then applying Claim~\ref{claim:prog-opt_2}, we arrive at a contradiction
to the claim that our initial solution was optimal.

Among the biases left fixed, let $n_1$ denote the number of Bernoullis with bias $1$, and $n_r$ denote the number with bias $1-r$. As discussed previously, we note $n_1\le k-2$. Let $k'=k-n_1\ge 2$ denote the number of item units left available once the $n_1$ Bernoullis with bias $1$ have each acquired an item. Recall that $Q_1$ is the probability that exactly $1$ unit is left available for the two Bernoullis from the others. Rephrasing, $Q_1$ is the probability that of the remaining Bernoullis, exactly $k-1$ take a non-zero value. Since $n_1$ Bernoullis take on a non-zero value with certainty, this means that among the $n_r$ Bernoullis with bias $1-r$, exactly $k'-1$ take on a non-zero value where $k'=k-n_1$. Likewise, $Q_2$ is the probability that among the $n_r$ Bernoullis with bias $1-r$, exactly $k'-2$ take on a non-zero value, and $Q_{\ge 3}$ is the probability that among the $n_r$ Bernoullis with bias $1-r$, at most $k'-3$ take on a non-zero value. Let us also define $Q_i$ for $i\ge 3$ as the probability that among the $n_r$ Bernoullis with bias $1-r$, exactly $k'-i$ take on a non-zero value. Observe that since $k'\ge 2$, we must have $Q_1, Q_2>0$.

We first claim that $k'\ge 3$. If not, then we have $Q_{\ge 3}=0$, which implies $(Q_2)^2-Q_{\ge 3}(Q_1-Q_2)>0$ and completes the proof. We can now compute the probabilities $Q_i$:
\begin{align*}
    Q_1={n_r-1 \choose k'-1} (1-r)^{k'-1}r^{n_r-k'}, \qquad Q_2={n_r-1 \choose k'-2}(1-r)^{k'-2}r^{n_r-k'+1}, \qquad etc.
\end{align*}

Let us denote by $\alpha$ the ratio between $Q_2$ and $Q_1$:
\begin{align*}
    \alpha & = \frac{Q_2}{Q_1} = \frac{k'-1}{n_r-k'+1} \cdot \frac{r}{1-r}
\end{align*}
We then observe:
\begin{align*}
    \frac{Q_3}{Q_2} & = \frac{k'-2}{n_r-k'+2}\cdot \frac{r}{1-r} <\alpha & \text{ and likewise,}\\
    \frac{Q_{i+1}}{Q_i} & \le \frac{k'-i}{n_r-k'+i} \cdot \frac{r}{1-r} <\alpha & \forall i> 3 \text{ with } Q_i>0
\end{align*}
Here the second inequality follows by noting that $\frac{k'-i}{n_r-k'+i}$ decreases with $i$.

We can therefore write $Q_{\ge 3} = \sum_{i\ge 3} Q_i< Q_1(\alpha^2+\alpha^3+\cdots) = \alpha^2 Q_1/(1-\alpha)$. Putting these expressions together we obtain:
\begin{align*}
    (Q_2)^2-Q_{\ge 3}(Q_1-Q_2) &> \alpha^2 Q_1^2 - \alpha^2 \frac{Q_1}{1-\alpha} (Q_1-\alpha Q_1) = 0.
\end{align*}
This completes the proof.
\end{proof}

\subsection{Concluding the proof of Lemma~\ref{lem:Poisson_worst_case}}
We are now ready to provide the proof of the main lemma of the section.
\begin{proof}[Proof of Lemma~\ref{lem:Poisson_worst_case}]
By Lemma~\ref{lem:bernoullis}, $\phi_k(\f)$ attains its lowest value when all the distributions are Bernoulli. By Lemma~\ref{lem:equal_or_degenerate}, all these Bernoulli distributions either have equal biases or are degenerate (have $0$ or $1$ bias). By Lemma~\ref{lem:equalbiases_new}, there exists an optimal solution with no degenerate Bernoulli distributions. If there are finite Bernoulli distributions, we can always repeat the argument of Lemma~\ref{lem:equalbiases_new} and obtain a solution with strictly higher objective and one more non-zero bias. As a result, there exists an optimal solution that consists of an infinite collection of Bernoulli distributions with the same bias. This establishes that the worst-case instance is the limit $n\rightarrow \infty$.

We now claim that the quantity $X_{\pi(k,\f)}$ in this limit is a Poisson distribution. Let $Y_n$ be the sum of $n$ equal Bernoulli distributions with $\mu_n:=\mu_n(Y_n)=\delta_n(Y_n)=:\delta_n$. Note that $\mu_n$ and $\delta_n$ are both decreasing functions of $n$ by the above argument and bounded in $[0,1]$. Therefore, they converge to a single limit, $\mu^\star=\delta^\star$. Let $\lambda_n$ denote $E[Y_n]$. We will show that the sequence $\{\lambda_n\}$ converges. Then, by the Poisson limit theorem, the limit of $\{Y_n\}$, namely $Y^\star=X_{\pi(k,\f)}$, is Poisson. This means that $\mu_n$ and $\delta_n$, which are linear functions of the pdf, also converge to $\mu(Y^{\star})$ and $\delta(Y^{\star})$ respectively. Because $\mu_n$ and $\delta_n$ are identical sequences,  $\mu(Y^{\star})=\delta(Y^{\star})$, which completes the proof. 

It remains to show that $\{\lambda_n\}$ converges; this does not follow immediately from the convergence of $\{\mu_n\}$ and $\{\delta_n\}$ because $\{\lambda_n\}$ may not evolve monotonically. We will show that for any $\epsilon>0$, there exists an $N(\epsilon)$ such that for any $n_1, n_2>N(\epsilon)$, $|\lambda_{n_1}-\lambda_{n_2}|\le\epsilon$, from which our claim follows.

In the remainder of this proof, we use $\text{Bin}(n, \lambda)$ to denote the sum of $n$ Bernoulli variables each with bias $\lambda/n$. We will use the following lemma that shows that two Binomials with sufficiently different expectations will also have sufficiently different truncated expectations. The proof of the lemma appears at the end of this section.

\begin{lemma}\label{lem:difference_Binomials}
For any $\epsilon>0$, any $\lambda, \lambda'\le k$ with $|\lambda-\lambda'|>\epsilon$, and any $n\geq 2k$, define $\epsilon':=(1-e^{-\epsilon/2})/4k$. Then, 
it holds that $|\mu\big(\text{Bin}(n,\lambda)\big)-\mu\big(\text{Bin}(n,\lambda')\big)|>2\epsilon'$.
\end{lemma}
We now make the following observations. 
\begin{enumerate}
    \item 
    Observe that for any $n>0$, $\delta_n\geq 1/2$, which means that $X_n\leq k-1$ with probability at least $1/2$ and the median is thus no more than $k-1$. Since the mean and the median of a Binomial differ by less than $\ln(2)$ \cite{Hamza1995}, this implies that $\lambda_n< k$ for all $n$ . 
    \item For any $\epsilon'>0$, using the fact that $\{\mu_n\}$ converges, there exists $N_1(\epsilon')$ such that for any $n_1, n_2>N_1(\epsilon')$, $|\mu_{n_1}-\mu_{n_2}|\le\epsilon'$, i.e., $|\mu(\text{Bin}(n_1, \lambda_{n_1}))-\mu(\text{Bin}(n_2, \lambda_{n_2}))| \le \epsilon'$. 
    \item Using the fact that the pdf of the Binomial distribution converges in $\ell_1$ norm as the number of samples increases while the bias stays the same, and that $\mu$ is a linear function of the pdf, we get that  for any $\epsilon'>0$, there exists $N_2(\epsilon')$, such that for every $\lambda\le k$, and every $n_1, n_2 >N_2(\epsilon')$, $|\mu(\text{Bin}(n_1, \lambda))-\mu(\text{Bin}(n_2, \lambda))|\le\epsilon'$. \footnote{This is immediate for a fixed $\lambda$ by the Poisson convergence theorem. We can then define $N_2(\epsilon')$ as the supremum over the corresponding $\lambda$-specific values. That this is finite can be seen by discretizing the set of all $\lambda$s.} 
    \item 
    Putting these together, for $n_1,n_2\geq \max\{N_1(\epsilon'), N_2(\epsilon')\}$: \begin{align*}|\mu(\text{Bin}(n_2, \lambda_{n_1}))-\mu(\text{Bin}(n_2, \lambda_{n_2}))| &\le |\mu(\text{Bin}(n_1, \lambda_{n_1}))-\mu(\text{Bin}(n_2, \lambda_{n_1}))| \\&\qquad+ |\mu(\text{Bin}(n_1, \lambda_{n_1}))-\mu(\text{Bin}(n_2, \lambda_{n_2}))| \le \epsilon'+\epsilon' = 2\epsilon'\end{align*}
    \item Finally, given $\epsilon>0$, we set $\epsilon'=(1-e^{-\epsilon/2})/4k$ and  $N(\epsilon)=\max\{N_1(\epsilon'), N_2(\epsilon'),2k\}$. Lemma~\ref{lem:difference_Binomials} then implies that for all $n_1, n_2\ge N(\epsilon)$,    $|\lambda_{n_1}-\lambda_{n_2}|\le\epsilon$, which concludes the proof. 
\end{enumerate}

\end{proof}

\begin{proof}[Proof of Lemma~\ref{lem:difference_Binomials}.] Without loss of generality, let $\lambda'>\lambda$ and let $b=\lambda/n$ and $b'=\lambda'/n$ be the bias of the Bernoulli variables constituting $\text{Bin}(n, \lambda)$ and $\text{Bin}(n, \lambda')$ respectively. We will make correlated draws from the two distributions and use these to bound the difference between their truncated expectations. We first flip $n$ coins with bias $b$ each to produce an instantiation of $\text{Bin}(n, \lambda)$; call this variable $X$. We then consider each of the $n$ coins that came up tails in the first experiment, and flip them again with bias $(b'-b)/(1-b)$ each, counting the total number of heads in the first and the second experiment together and denoting it $X'$; $X'$ is then an instantiation of $\text{Bin}(n, \lambda')$.

Now, we define $\mu :=\mu(\text{Bin}(n,\lambda)) = \text{E}[\min(X,k-1)]/k$ and $\mu' :=\mu(\text{Bin}(n,\lambda')) = \text{E}[\min(X',k-1)]/k$. Clearly, $\mu'>\mu$, and we can lower bound the difference between the two by $1/k$ times the probability of the event that $X<k-1$ and $X'>X$. Let us compute this probability.

First, we observe that $\text{E}[X]=\lambda\le k$, so $\text{Pr}[X\le k-1]\ge 1/2$. Then, conditioning on $X<k$ and recalling that $n\ge 2k$, we have at least $n/2$ coins flipped in the second experiment. The probability that all of the coins come up tails is at most $$\big(1-(b'-b)\big)^{n/2}\leq (1-\epsilon/n)^{n/2}\leq e^{-\epsilon/2}. $$

Therefore, conditioning on $X<k$, the event $X'>X$ happens with probability at least $1-e^{-\epsilon/2}$. Putting everything together, we get that $\mu'-\mu \ge (1-e^{-\epsilon/2})/2k$.
 \end{proof}

\section{Discussion}
\label{sec:tightness}
This paper provides an understanding of the performance of static pricing in multi-unit prophet inequalities. We show a simple static pricing scheme that obtains a competitive ratio that adapts to the size of the supply. This is enabled by a clean revenue-utility decomposition which allows to move beyond the now dominant \emph{balanced} approach to analyze prophet inequalities. We hope that this decomposition can further our understanding of prophet inequalities for more settings. In particular, it has already inspired  follow-up work as our technique has been a building block in providing multi-unit prophet inequality guarantees under the random-order arrival model~\cite{ArnostiMa2021}.

An interesting direction that remains open is to extend our approach to the setting with multiple different items. For its single-unit version, the existing guarantees rely on the~\emph{balanced prices} approach, which cannot yield competitive ratios beyond $1/2$ (see Section~\ref{ssec:related_work} for a relevant discussion). It would be interesting to provide improved guarantees for multi-item multi-unit prophet inequalities.

\subsection*{Acknowledgements}
The authors would like to thank Seffi Naor for multiple valuable discussions that helped shape the paper's contributions.

\bibliographystyle{alpha}
\bibliography{bibliog}

\end{document}